\documentclass[11pt]{article}
\usepackage[a4paper,lmargin=1in,rmargin=1in]{geometry}
\usepackage{amsmath}
\usepackage{amsfonts}
\usepackage{amsbsy}
\usepackage{amsthm}
\usepackage{amssymb}
\usepackage{accents}
\usepackage{tensor}
\usepackage{enumitem}
\usepackage{graphicx}
\newtheorem{thrm}{Theorem}

\newtheorem{prps}[thrm]{Proposition}
\theoremstyle{definition}
\newtheorem{rmrk}{Remark}

\DeclareMathOperator{\End}{End}

\DeclareMathOperator{\tr}{tr}
\DeclareMathOperator{\Tr}{Tr}
\DeclareMathOperator{\vol}{vol}
\DeclareMathOperator{\pf}{pf}
\DeclareMathAlphabet{\mathpzc}{OT1}{pzc}{m}{it}

\newcommand{\rmi}{\mathrm{i}}
\newcommand{\rmd}{\mathrm{d}}

\newcommand{\symped}[1]{\accentset{S}{#1}}

\newcommand{\sympman}{\mathcal{M}}
\newcommand{\bund}{\mathcal{E}}
\newcommand{\ebund}{{\End(\bund)}}

\newcommand{\connsymp}{\partial^{S}}
\newcommand{\connbund}{\partial^{\bund}}

\newcommand{\gambund}{\Gamma^{\bund}}

\newcommand{\curvbund}{R^{\bund}}
\newcommand{\curvsymp}{\symped{R}}

\newcommand{\sGamma}{\symped{\Gamma}}
\newcommand{\sympvol}{\vol_{S}}
\newcommand{\metvol}{\vol_{M}}

\newcommand{\slgr}{\mu}
\newcommand{\clgr}{\tau}
\newcommand{\wlgr}{\sigma}
\newcommand{\cslgr}{{\tau_\sharp}}
\newcommand{\shpv}{{V_\sharp}}

\begin{document}

\title{\textsc{Background independent noncommutative gravity from Fedosov quantization \\of endomorphism bundle}}
\date{}
\author{Micha{\l} Dobrski\footnote{michal.dobrski@p.lodz.pl}
\\
\small
\emph{Institute of Physics}
\\
\small
\emph{Lodz University of Technology}
\\
\small
\emph{W\'olcza\'nska 219, 90-924 {\L}\'od\'z, Poland}}
\maketitle
\abstract{
Model of noncommutative gravity is constructed by means of Fedosov deformation quantization of endomorphism bundle. The fields describing noncommutativity -- symplectic form and symplectic connection -- are dynamical, and the resulting theory is coordinate covariant and background independent. Its interpretation in terms of Seiberg-Witten map is provided. Also, new action for ordinary (commutative) general relativity is given, which in the present context appears as a commutative limit of noncommutative theory.
}

\section{Introduction}
Starting from the seminal paper by Seiberg and Witten \cite{seibwitt}, the great interest in noncommutative field theories with noncommutativity described by star products developed. Various theories of noncommutative general relativity (NCGR) have been studied in such context, and \cite{chamseddine1} - \cite{dimitri} provides far incomplete, although quite representative list of relevant works. With time, a tendency to make NCGR models more ``general relativistic'' can be observed, manifesting e.g. in efforts to restore its full coordinate covariance. The present paper stays within this trend and aims at introducing a variant of background independent geometric NCGR. This objective indicates that structures describing noncommutativity should be dynamical. In turn, one cannot persist in the simplest description by means of Moyal product, as it fixes constant noncommutativity tensor and breaks coordinate covariance. Although Moyal-based formulation can be geometrized \cite{aschieri1,aschieri2,aschieri3}, and noncommutativity can be given dynamics in such framework \cite{aschieri4}, rather different approach is going to be pursued here. Following the previous work \cite{dobrski-ncgr}, we choose Fedosov formalism as a main tool. 

The deformation quantization construction as given by Fedosov is considered to be ``very geometric'' due to formulation in the language of classical differential geometry -- all crucial steps are performed by analyzing connections and sections of certain bundles and the resulting structures are fully coordinate covariant. The key point of the approach introduced in \cite{dobrski-ncgr} is that it uses extension of Fedosov quantization to products of sections of appropriate endomorphism bundles. The obtained models enjoy some desirable properties -- they are fully geometric and global. It can be furthermore shown that they are locally equivalent to theories build by means of the variant of Seiberg-Witten map \cite{dobrski-sw,dobrski-ncgr}, hence the relation to more conventional strategies can be established. 

Yet, some difficulties with procedure proposed in \cite{dobrski-ncgr} can be easily spotted, and they essentially boil down to the presence of two, so far unrelated, geometries -- the metric and the symplectic ones. For example, when passing from commutative to noncommutative case, one replaces integral in the action by Fedosov trace functional. However, geometric action principles are usually  formulated by integrals involving metric volume form $\sqrt{|\det (g_{ab})|} \rmd x^1 \wedge \dots \wedge \rmd x^{2n}$, while the Fedosov trace functional is defined in terms of the symplectic one $\frac{\omega^n}{n!}$. Leaving metric and symplectic structures unrelated, leads to certain unnaturalness manifesting  in ``manual'' rescaling of the volume form, as in \cite{dobrski-ncgr}. On the other hand, it is not trivial to give physically plausible relationship between aforementioned structures. For instance, too strict dependence would result in unacceptable restrictions on the space-time geometry (compare remark \ref{strongrmrk} below). In the next section, we propose compatibility condition which seems to work reasonably in field-theoretic context. Using it, we are able to give new variational principle for general relativity, which appears to be suitably tailored for passing to noncommutative regime. In section \ref{sect3} a model of background independent, coordinate covariant, vacuum NCGR with dynamical noncommutativity is formulated and its interpretation in terms of Seiberg-Witten map is described. Perspectives on further research and other comments are provided in section \ref{sect4}. Some bulky formulas are moved to the closing appendix.

\section{Weak compatibility and related GR action}
\label{sect2}
We are going to investigate new variational principle yielding field equations of general relativity. From the purely classical perspective its form may seem rather unnatural, as we are introducing constrained system with new fields  peculiarly coupled to the metric structure. However, these additional degrees of freedom -- symplectic form and symplectic connection -- provide the input data required for the description of noncommutativity by the deformation quantization procedure. From this point of view, our action is dictated by noncommutativity of the spacetime, or more precisely (and more modestly) -- by its description by means of Fedosov construction.

The geometry which is going to be considered can be briefly described as that of \emph{weakly compatible metric Fedosov manifold}. The term ``Fedosov manifold'' refers to a $2n$ dimensional ($n>1$ in our case) symplectic manifold $(\sympman,\omega)$ equipped with a torsionfree symplectic connection $\connsymp$.  (For a discussion of geometry of Fedosov manifolds one may consult \cite{gelfand,bielgutt}). It moreover carries a~(pseudo-)Riemmaninan metric $g$. From now on we use the Latin indices as corresponding to some arbitrary coordinate basis of the tangent bundle. In general it is not assumed that theses coordinates are of the Darboux type, unless otherwise stated. The convention for connection coefficients $\tensor{\sGamma}{^i_{jk}}$ of the symplectic connection $\connsymp$ is settled by
\begin{equation}
\label{afcondef}
\connsymp_i X^j = \frac{\partial X^j}{\partial x^i}+ \tensor{\sGamma}{^j_{mi}} X^m
\end{equation}
where $X^j$ stands for components of tangent vector field. The components of the curvature tensor of $\connsymp$ are given by
\begin{equation}
\tensor{\curvsymp}{^i_j_k_l}=\frac{\partial \tensor{\sGamma}{^i_{jl}}}{\partial x^k} - \frac{\partial \tensor{\sGamma}{^i_{jk}}}{\partial x^l} + \tensor{\sGamma}{^i_{mk}} \tensor{\sGamma}{^m_{jl}} - \tensor{\sGamma}{^i_{ml}} \tensor{\sGamma}{^m_{jk}}
\end{equation}
The analogous relations hold for connection coefficients $\tensor{\Gamma}{^i_{jk}}$ of Levi-Civita connection $\nabla$ of the metric $g$, and for the Riemann curvature tensor $\tensor{R}{^i_{jkl}}$. The following condition is imposed to establish the relationship between metric and symplectic structures.
Define the tensor\footnote{Recall that difference of connection coefficients truly defines a tensor with proper transformation properties.}
\begin{equation}
\tensor{Q}{^i_{jk}}=\tensor{\Gamma}{^i_{jk}} -\tensor{\sGamma}{^i_{jk}}
\end{equation}
We are going to call $\connsymp$ and $\nabla$ \emph{weakly compatible} if $\tensor{Q}{^j_{ji}}=0$. There are some important consequences of this requirement.
\begin{thrm}
\label{detfactthrm}
The condition $\tensor{Q}{^j_{ji}}=0$ is locally equivalent to 
\begin{equation}
\label{frstdetsprop}
\det(\omega_{ab})=e^C |\det(g_{ab})|
\end{equation}
with $C \in \mathbb{R}$. 
\end{thrm}
\begin{proof}
Suppose that a connection $\partial^M$ preserves nondegenerate $2$-tensor $m_{ab}$ with inverse $M^{ab}$, i.e. $\partial^M_i m_{jk}=0$. It follows that 
\begin{equation}
\label{prep1f1}
\frac{1}{2} M^{k j} \frac{\partial m_{jk}}{\partial x^i} = \tensor{\accentset{M}{\Gamma}}{^j_{ji}}
\end{equation}
with $\tensor{\accentset{M}{\Gamma}}{^i_{jk}}$ denoting connection coefficients of $\partial^M$ defined as in (\ref{afcondef}). On the other hand the formula for partial derivatives of a determinant yields
\begin{equation}
\label{prep1f2}
M^{k j} \frac{\partial m_{jk}}{\partial x^i} = \frac{1}{\det (m_{ab})} \frac{\partial \det (m_{ab})}{ \partial x^i}
\end{equation}
Combining (\ref{prep1f1}) with (\ref{prep1f2}) one obtains
\begin{equation}
\label{gammtodet}
\tensor{\accentset{M}{\Gamma}}{^j_{ji}}=\frac{1}{2 \det (m_{ab})} \frac{\partial \det (m_{ab})}{ \partial x^i}
\end{equation}
Applying above formula to $\connsymp$ and $\nabla$ in the weak compatibility condition $\tensor{Q}{^j_{ji}}=0$ gives the equation
\begin{equation}
\label{detdiffeq}
\frac{1}{\det (\omega_{ab})} \frac{\partial \det (\omega_{ab})}{ \partial x^i}=\frac{1}{\det (g_{ab})} \frac{\partial \det (g_{ab})}{ \partial x^i}
\end{equation}
which, in turn, produces (\ref{frstdetsprop}). Conversely, differentiating (\ref{frstdetsprop}) and eliminating $e^C$ from the result one arrives at (\ref{detdiffeq}) from which $\tensor{Q}{^j_{ji}}=0$ follows by (\ref{gammtodet}).
\end{proof}
As the immediate consequence we obtain the following result.
\begin{prps}
\label{propdets}
The symplectic volume form $\sympvol=\frac{\omega^n}{n!}$ and the metric volume form $\metvol=\sqrt{|\det(g_{ab})|} \rmd x^1 \wedge \dots \wedge \rmd x^{2n}$ are proportional
\begin{equation}
\label{volsprop}
\sympvol=\alpha \metvol
\end{equation}
with constant $\alpha > 0$, on each connected component of $\sympman$.
\end{prps}
Indeed, the relation (\ref{volsprop}) holds, as the consequence of (\ref{frstdetsprop}) together with the definition and properties of the Pfaffian, i.e. due to 
\begin{equation}
\frac{\omega^n}{n!}=\pf(\omega_{ab})\rmd x^1 \wedge \dots \wedge \rmd x^{2n}= \sqrt{\det(\omega_{ab})} \rmd x^1 \wedge \dots \wedge \rmd x^{2n}
\end{equation}

\begin{prps}
\label{propdivs}
The covariant divergences of polyvectors coincide for $\connsymp$ and $\nabla$.
\begin{equation}
\label{divsprop}
\connsymp_{i_1} X^{i_1 \dots i_k} = \nabla_{i_1} X^{i_1 \dots i_k}
\end{equation}
\end{prps}
The proof is a straightforward calculation employing weak compatibility condition and antisymmetry of $X^{i_1 \dots i_k}$.

\begin{rmrk}
\label{localrmrk}
Given a metric, one can always find a weakly compatible Fedosov structure locally. Indeed, this can be achieved by taking some unimodular coordinates $|\det(g_{ab})|=1$ as Darboux coordinates of $\omega_{ij}$ for which symplectic connection coefficients are given by $\tensor{\sGamma}{^i_{jk}}=\Lambda^{il}A_{ljk}$ with arbitrary completely symmetric $A_{ljk}$ and $\Lambda^{ij}$ denoting inverese of $\omega_{ij}$. Obstructions may nevertheless appear for the global setting, e.g. topology forbidding existence of any symplectic form.   
\end{rmrk}
\begin{rmrk}
\label{strongrmrk}
One could consider the stronger compatibility condition $\tensor{Q}{^i_j_k}=0$, which means that the connections $\connsymp$ and $\nabla$ coincide. This however would be too constraining for our purposes. Indeed, in dimension $4$ the existence of $2$-form constant with respect to the Levi-Civita connection implies that the underlying metric is decomposable, i.e. it locally can be represented as a sum of two $2$-dimensional metrics \cite{stephani2003exact}. Moreover, if one demands the Ricci-flatness, then there is a very narrow class of metrics satisfying such requirement in the Lorentzian case \cite{petrov1966}. This observation prevents us from imposing such condition in the context of NCGR model building.
\end{rmrk}
We can now formulate our variational principle. Suppose that $\mathcal{L}_m$ is a matter Lagrangian, which in the standard approach is\ integrated with the canonical metric volume form and contributes to the action by $\int_\sympman \mathcal{L}_m \metvol$. Also, let $\tensor{R}{^i_{jkl}}$ denote Riemann curvature tensor of $\nabla$, $R_{ij}=\tensor{R}{^k_{ikj}}$ -- its Ricci tensor and $R=g^{ij}R_{ij}$ -- its Ricci scalar. Consider an action integral 
\begin{equation}
\label{graction}
S=\frac{1}{16 \pi G} \int_\sympman  R \frac{\omega^n}{n!} + \int_\sympman \mathcal{L}_m  \frac{\omega^n}{n!}
\end{equation}
with dynamical variables determined by weakly compatible metric Fedosov structure, i.e. by a metric $g$, a symplectic form $\omega$ and coefficients of a symplectic connection $\connsymp$ restricted by the constraints
\begin{subequations}
\label{cnstr}
\begin{align}
\label{cnstrsymp}
\connsymp_i \omega_{jk} &=0\\
\label{cnstrtorsion}
\tensor{\sGamma}{^i_{[jk]}} &=0\\
\label{cnstrcomp}
\tensor{\Gamma}{^j_{ji}} -\tensor{\sGamma}{^j_{ji}} &=0
\end{align}
\end{subequations}
which are, in the case of (\ref{cnstrsymp}) and (\ref{cnstrcomp}), nonholonomic. Observe, that from (\ref{cnstrsymp}) and vanishing of torsion expressed by (\ref{cnstrtorsion}), it follows that $\rmd \omega=0$, thus the symplectic condition is encoded in the constraints.  Furthermore the action is diffeomorphism invariant both in passive and active sense. The former symmetry follows from the coordinate covariance. The latter, from the fact that all geometric objects entering (\ref{graction}) are subject to variations and also constraints are preserved by active diffeomorphisms.

Let us use Lagrange multiplier method for finding field equations. Introduce tensors $\slgr^{ijk}$, $\tensor{\clgr}{_i^{jk}}$ and $\wlgr^i$ with symmetry properties $\slgr^{ijk}=-\slgr^{ikj}$, $\tensor{\clgr}{_i^{jk}}=-\tensor{\clgr}{_i^{kj}}$.  Consider auxiliary unconstrained action integral
\begin{equation}
S_{aux}=S + \int_\sympman \left( \slgr^{ijk}\connsymp_i \omega_{jk} + \tensor{\clgr}{_i^{jk}} \tensor{\sGamma}{^i_{[jk]}} + \wlgr^i (\tensor{\Gamma}{^j_{ji}} -\tensor{\sGamma}{^j_{ji}}) \right)  \frac{\omega^n}{n!}
\end{equation}
Obviously, variations of $S_{aux}$ with respect to Lagrange multipliers produce constraint equations (\ref{cnstr}). Variations with respect to $g_{ij}$, $\tensor{\sGamma}{^k_{ij}}$ and $\omega_{ij}$, after employing proposition \ref{propdets}, give respectively 
\begin{subequations}
\label{fldeq}
\begin{align}
\label{fldeqg}
\frac{1}{16 \pi G} R^{ij}+\frac{1}{2}g^{ij}\nabla_k \wlgr^k & =\frac{1}{2}\symped{T}^{ij}\\
\label{fldeqgam}
\tensor{\delta}{_k^i} \wlgr^j +2 \slgr^{jil}\omega_{kl} &=\tensor{\clgr}{_k^{ij}}\\
\label{fldeqom}
\frac{1}{2}\left(\frac{1}{16 \pi G} R +\mathcal{L}_m\right)\Lambda^{ji}&=\connsymp_k \slgr^{kij}
\end{align}
\end{subequations}
where $\Lambda^{ij}$ is the inverse of $\omega_{ij}$ and
\begin{equation}
\symped{T}^{ij}=2\frac{\partial{\mathcal{L}_m}}{\partial g_{ij}}-\frac{2}{\sqrt{\det(\omega_{ab})}} \frac{\partial}{\partial x^m}
\left( \sqrt{\det(\omega_{ab})} \frac{\partial \mathcal{L}_m}{\partial \left( \frac{\partial g_{ij}}{\partial x^m}\right)} \right)
\end{equation} 
can be related due to proposition \ref{propdets} to the standard energy-momentum
tensor $T^{ij}$ by
\begin{equation}
\label{serelation}
T^{ij}=\symped{T}^{ij}+g^{ij} \mathcal{L}_m
\end{equation}
On the other hand, variations with respect to matter fields entering into $\mathcal{L}_m$ yield, again with the help of proposition \ref{propdets}, the same equations of motion as that originating from the usual action for matter $\int_\sympman \mathcal{L}_m  \metvol$. Then, assuming diffeomorphism invariance of the matter action in (\ref{graction}), one recovers, by the straightforward modification of the standard calculation, the conservation of energy-momentum tensor 
\begin{equation}
\label{emconserv}
\nabla_i T^{ij}=0
\end{equation} 
Computing $\connsymp_k \wlgr^k$ from (\ref{fldeqgam}) and (\ref{fldeqom}), using proposition \ref{propdivs} and substituting the result into (\ref{fldeqg}) one obtains
\begin{equation}
\label{enstfeq}
\frac{1}{8 \pi G} \left( R^{ij} -\frac{1}{2}g^{ij} R \right)+  g^{ij} \lambda =  T^{ij}
\end{equation}
with $\lambda=\frac{\connsymp_k \tensor{\clgr}{_l^{lk}}}{2n}$. Taking $\nabla_i$ of both sides of (\ref{enstfeq}) shows that $\lambda=\mathrm{const}$ due to Bianchi identities and (\ref{emconserv}). Thus, we have proved the following result.

\begin{prps}
\label{prpsA}
Each solution of field equations (\ref{fldeq}) must include metric satisfying Einstein equations with cosmological constant (\ref{enstfeq}). 
\end{prps}

Let us check if the converse statement holds true. For this purpose it is convenient to transform equation for Lagrange multipliers in the following manner. From (\ref{enstfeq}) it can be easily calculated that 
\begin{equation}
\label{rformula}
R=\frac{8 \pi G (T-2n \lambda)}{1-n}
\end{equation}
for $T=g_{ij}T^{ij}$. Introduce $\cslgr^{ijk}:=\Lambda^{il}\tensor{\clgr}{_l^{jk}}$ and rewrite (\ref{fldeqgam}) as
\begin{subequations}
\label{lagreqs}
\begin{equation}
\label{fldeqgamsharp}
\slgr^{jik}=\frac{1}{2}\left(\,\,\cslgr^{kij}-\wlgr^j \Lambda^{ki}\right)
\end{equation}
It follows immediately from this formula that $\cslgr^{kij}=-\cslgr^{ikj}$, i.e. $\cslgr^{kij}$ is completely antisymmetric. Using relation (\ref{rformula}) we infer from (\ref{fldeqg})
\begin{equation}
\label{lagrweq}
\nabla_k \wlgr^k = \frac{\lambda}{1-n}-\mathcal{L}_m -\frac{T}{2-2n}
\end{equation}
Employing above formula in (\ref{fldeqgamsharp}) and (\ref{fldeqom}) yields, by propostion \ref{propdivs}
\begin{equation}
\label{lagrcshrpeq}
\nabla_k \cslgr^{ijk}=\connsymp_k \,\cslgr^{ijk}=\lambda \Lambda^{ij}
\end{equation}
\end{subequations}
Observe that from equations (\ref{lagreqs}) and (\ref{enstfeq}) one can return back to (\ref{fldeq}). In fact, (\ref{fldeqgam}) results trivially form (\ref{fldeqgamsharp}). After multiplying (\ref{lagrweq}) by $g^{ij}$ and using (\ref{rformula}), (\ref{enstfeq}) and (\ref{serelation}) one obtains (\ref{fldeqg}). Finally, the relation (\ref{fldeqom}) can be derived by applying $\connsymp_j$ to (\ref{fldeqgamsharp}) and utilizing (\ref{lagrcshrpeq}), (\ref{lagrweq}) together with (\ref{rformula}). Hence, equations (\ref{lagreqs}) and (\ref{enstfeq}) are equivalent to (\ref{fldeq}).

Notice that given a metric satisfying Einstein field equations (\ref{enstfeq}) one can, by remark \ref{localrmrk}, construct weakly compatible Fedosov structure locally. Then, equations (\ref{lagreqs}) can be solved, yielding Lagrange multipliers satisfying (\ref{fldeq}). Indeed equation (\ref{lagrcshrpeq}) in Darboux coordinates takes form
\begin{equation}
\frac{\partial \cslgr^{ijk}}{\partial x^k} = \lambda \Lambda^{ij}
\end{equation}
and is solved by
\begin{equation}
\cslgr^{ijk} = \frac{3 \lambda}{2(n-1)}  \Lambda^{[ij} x^{k]} 
\end{equation}
Similarly, if we put $f(x)$ for r{.}h{.}s{.} of (\ref{lagrweq}) then in Darboux coordinates it reduces to $\frac{\partial \wlgr^k}{\partial x^k}=f(x)$ and one can easily construct a solution, e.g. by choosing $\wlgr^1 = \int f(x) dx^1$ and $\wlgr^2= \dots =\wlgr^{2n}=0$. Then $\slgr^{jik}$ can be algebraically computed from (\ref{fldeqgamsharp}). The following proposition has been therefore proved.

\begin{prps}
\label{prpsB}
Locally, for each solution of Einstein equations (\ref{enstfeq}) there exist weakly compatible Fedosov structure and Lagrange multipliers satisfying field equations (\ref{fldeq}).
\end{prps}
Propositions \ref{prpsA} and \ref{prpsB} ensure that the theory considered in this section is locally equivalent to the classical general relativity with cosmological constant. However, one should be aware of possible global issues, as already mentioned in remark \ref{localrmrk}.

Finally, let us comment on the symplectic content of our model. Equations (\ref{cnstrsymp}) and (\ref{cnstrtorsion}) show that $\omega$ together with $\connsymp$ constitute Fedosov manifold. The only further restriction is given by weak compatibility (\ref{cnstrcomp}). As it follows from the theorem \ref{detfactthrm}, this is equivalent to coupling $g$ with $\omega$ by determinants only. Thus, all weakly compatible Fedosov structures are gauge equivalent in the present setting.

\section{Noncommutative gravity}
\label{sect3}
A model of coordinate covariant, background independent noncommutative gravity which reduces to the variant of gravity described in the previous section in the commutative limit is going to be introduced now. For sake of simplicity we confine ourselves to the vacuum case of $\mathcal{L}_m \equiv 0$. The procedure described in \cite{dobrski-ncgr} will be followed. However, there are two notable improvements as compared to \cite{dobrski-ncgr}. First, because of the form of the action (\ref{graction}) we are no longer dealing with incompatibility of the volume forms -- the Fedosov trace functional is also built with the symplectic one. Second, as we switch from a fixed symplectic form and a fixed symplectic connection used in \cite{dobrski-ncgr} to the fully dynamical setting now, the theory becomes background independent.

\subsection{Preliminaries on Fedosov construction}
Our tool for deforming gravity into noncommutative theory is the Fedosov construction of deformation quantization of endomorphism bundles. For the convenience of the reader, key facts about  this formalism are presented here briefly, with all ``internal'' details omitted. If interested in them, one should consult \cite{fedosovart, fedosovbook} for beautiful original exposition. Further studies on geometric, algebraic and formal structure of the theory, as well as some examples, can be found in e.g. \cite{emmrwein,tosiek,waldmann}.

The arena for Fedosov quantization is given exactly by a Fedosov manifold as introduced in the previous section. The primary result is  explicit (although iterative) construction of global, associative, geometric formal star product of functions on $\sympman$. For $\chi$ denoting formal parameter, some initial terms of this product read
\begin{multline}
f*_S g=fg-\frac{\rmi \chi}{2} \connsymp_{i} f \, \Lambda^{ij}  \connsymp_{j} g
-\frac{\chi^2}{8} \connsymp_{(i}\connsymp_{j)} f \, \Lambda^{ik} \Lambda^{jl} \connsymp_{(k}\connsymp_{l)} g \\
+\frac{\rmi \chi^3}{48} \left(\connsymp_{(i} \connsymp_j \connsymp_{k)} f  -\frac{1}{4} \curvsymp_{(ijk)u} \Lambda^{uv} \connsymp_v f \right) \Lambda^{ip} \Lambda^{jr} \Lambda^{ks} 
 \left(\connsymp_{(p} \connsymp_r \connsymp_{s)} g  -\frac{1}{4} \curvsymp_{(prs)t} \Lambda^{tw} \connsymp_w gs \right)\\
+O(\chi^4)
\end{multline}
where $f,g \in C^{\infty}(\sympman)[[\chi]]$, $\tensor{\curvsymp}{^i_j_k_m}$ is the curvature tensor of symplectic connection and standard notation for symmetrization is used.

Now, let $\bund$ be a vector bundle over $\sympman$ and let $\ebund$ be corresponding endomorphism bunle, i.e. each fiber $\ebund_x$ consists of endomorphisms of respective fiber $\bund_x$. $\ebund$ comes with natural product of its sections which is noncommutative from the beginning (locally, for a fixed frame in $\bund$, it is just matrix multiplication). Fedosov construction provides global, geometric, associative deformation of this product into star product. Suppose that some connection $\connbund$ in 
$\bund$ is chosen. Let $\partial$ denote connection on arbitrary tensor product of $T\sympman$, $T^* \sympman$, $\bund$ and $\bund^*$ which combines $\connsymp$ and $\connbund$ (e.g. $\partial=\connsymp \otimes 1 + 1 \otimes \connbund$ for $T\sympman \otimes \bund$). Let $\gambund_i$ be local section of $\ebund$ corresponding to connection coefficients of $\connbund$ for some local frame and coordinates (i.e. $\connbund_i = \frac{\partial}{\partial x^i}+\gambund_i$) and let $\curvbund_{ab}=\frac{\partial}{\partial x^a} \gambund_b-\frac{\partial}{\partial x^b} \gambund_a+[\gambund_a,\gambund_b]$ be the curvature of $\connbund$. Then, Fedosov star product of endomorphisms reads 
\begin{multline}
\label{fedosov_endprod}
F*G=FG-\frac{\rmi \chi}{2} \Lambda^{a b} \partial_a F \partial_b G+\\
-\frac{\chi^2}{8}\Lambda^{ab}\Lambda^{cd}\Big(\{\partial_b F,\curvbund_{ac}\}\partial_d G + \partial_b F \{\curvbund_{ac},\partial_d G\} +\partial_{(a} \partial_{c)} F \partial_{(b} \partial_{d)} G\Big) + O(\chi^3)
\end{multline}
with $F,G \in C^\infty(\ebund)[[\chi]]$ and $\{\cdot , \cdot\}$ standing for the anticommutator.

For some special cases product (\ref{fedosov_endprod}) can be expressed in terms of star product of functions. If $\connbund$ is flat and local frame with $\gambund_i \equiv 0$ is fixed, then (\ref{fedosov_endprod}) reduces to a product of matrices with commutative multiplication of entries replaced by the noncommutative Fedosov star product of functions. We shall consequently use symbol $*_S$ for such case. One can moreover demand that $\connsymp$ is also flat and fix Darboux coordinates with vanishing coefficients $\tensor{\sGamma}{^i_{jk}} \equiv 0$. Then, Fedosov star product of functions reduces to Moyal product $*_M$ and (\ref{fedosov_endprod}) becomes multiplication widely used for the description of noncommutative gauge theories.

As we have seen, the input data for the construction of a star product of endomorphisms are: a symplectic form, a symplectic connection and a connection in $\bund$. Thus, on the same symplectic manifold one may construct different star products originating in different choices of $\connsymp$ and $\connbund$. Nevertheless, it can be shown that these star products are star equivalent (i.e. isomorphic) in the following sense. Let $*_1$ and $*_2$ be two star products obtained due to altered choice of connections. Then, there exists mapping $M=1+\chi M_1 +\chi^2 M_2 + \dots$, where $M_i$ are some differential operators, such that\footnote{More strictly, there are also other, internal degrees of freedom in Fedosov theory (and we silently assigned to them some canonical values), which, when modified, can produce inequivalent star products for certain nontrivial topologies.}
\begin{equation}
M(F *_1 G)=M(F) *_2 M(G)
\end{equation} 

The key object for our approach is Fedosov trace functional which realizes noncommutative variant of integration over manifold. For each Fedosov star product $*$ one is able to construct trace $\tr_*$ defined for compactly supported elements of $C^{\infty}(\ebund)[[\chi]]$, taking values in $\mathbb{C}[[\chi]]$, satisfying
\begin{equation}
\label{trcprop}
\tr_*(F*G)=\tr_*(G*F)
\end{equation} 
 and being invariant on star equivalences
\begin{equation}
\label{trinvprop}
\tr_{*_1} F = \tr_{*_2} M(F)
\end{equation}
It turns out that these requirements define trace completely up to normalizing constant. This fact is rooted in the observation that for the Moyal product the trace is just integral
\begin{equation}
\tr_{*_M} F = \mathrm{const} \int_{\mathbb{R}^{2n}} \Tr (F) \frac{\omega^n}{n!}
\end{equation} 
where $\Tr$ is the pointwise trace of an endomorphism. For arbitrary star product trace can be computed either by considering local trivializations to Moyal product or using different methods \cite{fedosov_on_the_trace,dobrski-ncgr}. The result reads
\begin{equation}
\label{fedosov_trfun}
\tr_* (F)=\int_{\sympman} \Tr\Bigg(F + \frac{\rmi \chi }{2}\Lambda^{a b}\curvbund_{a b} F 
 +\chi^2 \bigg(-\frac{3}{8} \Lambda^{[a b} \Lambda^{c d]}\curvbund_{a b} \curvbund_{c d} +s_2 \bigg)F + O(\chi^3) \Bigg) \frac{\omega^n}{n!}
\end {equation}
with
\begin{equation*}
s_2=\frac{1}{64} \Lambda^{[a b} \Lambda^{c d]}\tensor{\curvsymp}{^k_{lab}}\tensor{\curvsymp}{^l_{kcd}}  + \frac{1}{48}\Lambda^{ab}\Lambda^{cd}\connsymp_e \connsymp_a \tensor{\curvsymp}{^e_{bcd}}
\end{equation*}

\subsection{Action and field equations}
Let us deform action (\ref{graction}) using scheme developed in \cite{dobrski-ncgr}. The main idea is to interpret Lagrangian as a pointwise trace of some endomorphism, and then to construct deformed action by replacing integral with the trace functional. (As we shall see later, this procedure indeed introduces noncommutativity of the spacetime). Hence, the bundle $\bund$ and connection $\connbund$ must be specified first. Let $\bund=T \sympman$ and $\connbund=\nabla$. 
Furthermore the section of endomorphism bundle $\End (T\sympman)$ is needed. From the commutative limit (\ref{graction}) one infers that pointwise trace of this endomorphism should yield Ricci scalar. Then, the obvious natural candidate is Ricci tensor with the first index raised. Thus, we define endomorphism $\underline{R} \in C^{\infty}(\End (T\sympman))$ by setting its local components to $\tensor{R}{^i_j}$, or to state the same differently, by defining global action of $\underline{R}$ on tangent vector field $X \in C^\infty (T\sympman)$ with the local formula $(\underline{R} X)^i=\tensor{R}{^i_j}X^j$.
The noncommutative action which deforms action (\ref{graction}) reads 
\begin{equation}
\label{ncactn}
\widehat{S}=\frac{1}{16 \pi G} \tr_* (\underline{R}) 
\end{equation}
and is taken together with constraints (\ref{cnstr}). Notice that (\ref{cnstrsymp}) and (\ref{cnstrtorsion}) are conditions which ensure consistency of Fedosov quantization. In this way the variational procedure can be interpreted now as ``variation over Fedosov deformation quantizations''. The action (\ref{ncactn}) is diffeomormphism invariant in the very same way as (\ref{graction}), since we neither break coordinate covariance, nor introduce new fields in (\ref{ncactn}). For sake of further analysis let us break (\ref{ncactn}) into two parts
\begin{equation}
\widehat{S}=S + S_{c}
\end{equation}
where $S$ is exactly the action of (\ref{graction}) with $\mathcal{L}_m \equiv 0$ and $S_c$ corresponds to noncommutative corrections (compare with appendix for specific formula up to $\chi^2$).

Again, an auxiliary unconstrained action $\widehat{S}_{aux}$ with same Lagrange multipliers can be introduced. The field equations following from it take form 
\begin{subequations}
\label{ncfldeq}
\begin{align}
\label{ncfldeqg}
\frac{1}{16 \pi G} R^{ij}+\frac{1}{2}g^{ij}\nabla_k \wlgr^k & =W^{ij}\\
\label{ncfldeqgam}
\tensor{\delta}{_k^i} \wlgr^j +2 \slgr^{jil}\omega_{kl} +\tensor{V}{_k^{ij}}&=\tensor{\clgr}{_k^{ij}}\\
\label{ncfldeqom}
\frac{1}{32 \pi G} R \Lambda^{ji}+U^{ij}&=\connsymp_k \slgr^{kij}
\end{align}
\end{subequations}
and are supplemented by constraints (\ref{cnstr}). Here, $W^{ij}$, $\tensor{V}{_k^{ij}}$ and $U^{ij}$ are tensors being power series in deformation parameter $\chi$, whose first terms of $\chi^2$ order are given in the appendix. $W^{ij}$ is symmetric, while $U^{ij}$ is antisymmetric, because they are produced by variations of $g$ and $\omega$ respectively. 

Let us analyze strategy for solving Lagrange multipliers in the field equations. Essentially, one may proceed as in the commutative case. However, the nontrivial integrability conditions appear now, yielding some new relations between dynamical variables. First, (\ref{ncfldeqgam}) and (\ref{ncfldeqom}) can be used to determine $\connsymp_k \wlgr^k$. Defining
\begin{equation}
p(x):=\frac{1}{2n}\Big(\connsymp_j \tensor{\clgr}{_k^k^j}-2\omega_{kl}U^{kl}-\connsymp_j \tensor{V}{_k^k^j}\Big)
\end{equation}
one obtains 
\begin{equation}
\label{ncwlgreq}
\connsymp_k \wlgr^k=p(x)-\frac{1}{16 \pi G}R
\end{equation}
and it is possible to rewrite (\ref{ncfldeqg}) as
\begin{subequations}
\label{ncfldeqs}
\begin{equation}
\label{ncensteq}
\frac{1}{8 \pi G} \Bigg( R^{ij}-\frac{1}{2}g^{ij} R \Bigg)+g^{ij} p(x) = 2 W^{ij}
\end{equation}
Then, by contracted Bianchi identities we get linear inhomogeneous equation
\begin{equation}
\label{peq}
\nabla_k p(x)=2 g_{jk} \nabla_i W^{ij}
\end{equation}
for $p(x)$. Its general solution can be written as $p(x)=\lambda+p_0(x)$ where $\lambda$ is a (cosmological) constant and $p_0 (x)$ is arbitrary special solution of (\ref{peq}), which without loss of generality can be assumed to be of $\chi^2$ order, as $W^{ij}$ is. The integrability condition for (\ref{peq}) can be easily obtained and it reads
\begin{equation}
\label{peqcomp}
g_{jk}\nabla_m \nabla_i W^{ij}= g_{jm}\nabla_k \nabla_i W^{ij}
\end{equation}
\end{subequations}
Remarkably, it depends on $g$, $\omega$ and $\connsymp$ only. Defining $\shpv^{kij}:=\Lambda^{km}\tensor{V}{_m^i^j}$ one can rewrite (\ref{ncfldeqgam}) as
\begin{equation}
\label{ncgensharp}
\slgr^{jik}=\frac{1}{2}\left(\cslgr^{kij}-\wlgr^j \Lambda^{ki}- \shpv^{kij}\right)
\end{equation}
Notice, that  antisymmetry of $\cslgr^{kij}$ in first two indices cannot be inferred this time, instead one observes that
\begin{equation}
\label{nctshrp}
\cslgr^{(ki)j}=\shpv^{(ki)j}
\end{equation} 
Using (\ref{ncgensharp}) in equation (\ref{ncfldeqom}) and applying formula for $\connsymp_k \wlgr^k$ as before, one obtains
\begin{equation}
\label{nclagrcshrpeq}
\connsymp_k \cslgr^{ijk}=p(x) \Lambda^{ij} - 2U^{ij}+\connsymp_k \shpv^{ijk}
\end{equation}
and this equation requires some integrability check. Indeed, calculation of $\connsymp_j \connsymp_k \cslgr^{ijk}$ from (\ref{nclagrcshrpeq}), use of antisymmetry of $\cslgr^{ijk}$ in last two indices, and symmetry of symplectic Ricci tensor yields relation
\begin{equation}
\label{nclagrcshrpeqcomp}
\tensor{\curvsymp}{^i_j_k_m} \cslgr^{jkm}=2 \connsymp_j \Big( p(x) \Lambda^{ij} - 2U^{ij}+\connsymp_k \shpv^{ijk} \Big)
\end{equation}
which appears as a nontrivial condition\footnote{This relation is satisfied trivially in the commutative case due to the full antisymmetry of $\cslgr^{jkm}$ and constancy of $p(x)$.}. However, one can proceed as follows. Using diffeomorphism invariance of $S_c$ it is possible to derive the following identity\footnote{We omit detailed derivation of (\ref{ncconsq}) as it is straightforward analogue of the standard textbook derivation of $\nabla_i T^{ij}=0$ from the diffeomorphism invariance of the action for matter fields.}
\begin{equation}
\label{ncconsq}
2 g_{im} \nabla_j W^{ij}+2 \omega_{m j} \connsymp_i U^{ij} + \connsymp_i \connsymp_j \tensor{V}{_m^{ij}} - \tensor{\curvsymp}{^i_{jkm}} \tensor{V}{_i^{jk}}=0
\end{equation}
Employing (\ref{ncconsq}) together with (\ref{peq}) in (\ref{nclagrcshrpeqcomp}) simplifies it to
\begin{equation}
\label{intchkst2}
\tensor{\curvsymp}{^i_j_k_m} \cslgr^{jkm}=-2 \Lambda^{il} \omega_{js} \tensor{\curvsymp}{^s_{kml}} \shpv^{jkm}
\end{equation}
Now it is convenient to introduce $\curvsymp^{\flat}_{ijkl}:=\omega_{im} \tensor{\curvsymp}{^m_{jkl}}$. Condition (\ref{intchkst2}) can be  then rewritten as 
\begin{equation}
\label{intchkst3}
\curvsymp^{\flat}_{ijkm} \cslgr^{jkm}=-2  \curvsymp^{\flat}_{jkmi} \shpv^{jkm}
\end{equation}
On the other hand, the following symmetry for $\curvsymp^{\flat}_{ijkl}$ holds\footnote{Compare \cite{gelfand} for systamtic study of the structure of $\curvsymp^{\flat}_{ijkl}$. To obtain (\ref{cfsym}) one may use first Bianchi identities together with $\curvsymp^{\flat}_{ijkl}+\curvsymp^{\flat}_{lijk}+\curvsymp^{\flat}_{klij}+\curvsymp^{\flat}_{jkli}=0$. }
\begin{equation}
\label{cfsym}
\curvsymp^{\flat}_{ijkm}=\curvsymp^{\flat}_{jmki}-\curvsymp^{\flat}_{jkmi}
\end{equation}
and consequently, due to $\cslgr^{jkm}=-\cslgr^{jmk}$
\begin{equation}
\label{intchkst4}
\curvsymp^{\flat}_{jkmi} \cslgr^{jkm}=\curvsymp^{\flat}_{jkmi} \shpv^{jkm}
\end{equation}
But above equation is trivially satisfied because of $\curvsymp^{\flat}_{jkmi}=\curvsymp^{\flat}_{kjmi}$ and (\ref{nctshrp}). Thus, condition (\ref{nclagrcshrpeqcomp}) does not pose an obstruction for integrability of (\ref{nclagrcshrpeq}).

Above analysis provide a strategy for elimination of Lagrange multipliers from the field equations. Indeed, one starts with pointing out a solution of (\ref{ncfldeqs}) for $g$, $\omega$, $\connsymp$ and $p(x)$. Then $\cslgr^{ijk}$ can be determined from (\ref{nclagrcshrpeq}), and consequently $\wlgr^i$ and $\slgr^{ijk}$ can be obtained form (\ref{ncwlgreq}) and (\ref{ncgensharp}) respectively.
For now it is not clear to what extent equations (\ref{ncfldeqs}) fix symplectic data. Hopefully further perturbative expansion of the field equations could provide some conclusive results.

It is moreover substantial to verify if the considered theory contains imaginary terms in the action and in the field equations. The reality of action integral $\widehat{S}$ can be analyzed using method developed in \cite{dobrski-ncgrreal}. In fact, it was already checked there that endomorphism $\underline{R}$ is self-adjoint\footnote{Strictly, the self-adjointness was verified for the endomorphism proportional to $\underline{R}$ by the real scalar function. This of course implies the same property for $\underline{R}$.} with respect to the involution defined by the extension of metric $g$ to Hermitian metric on the complexification of $T\sympman$. Due to properties of the trace functional (compare theorem 1 in \cite{dobrski-ncgrreal}) one immediately obtains reality of our action integral, and in turn, of the field equations. Consequently imaginary terms indeed cannot appear in the theory. In particular, this result explains the vanishing of the first order correction, as the first order term in the general formula (\ref{fedosov_trfun}) for the trace is imaginary. Such observation is of course fully consistent with the well known phenomenon of lack of first order terms in NCGR theories obtained by the direct application of Seiberg-Witten map (compare e.g. \cite{mukherjee}).

\subsection{Seiberg-Witten map}
Let us investigate local noncommutative gauge symmetry of the proposed theory by interpreting it in terms of Seiberg-Witten map. This would also clarify the  relation of the present model to noncommutativity of the spacetime, as well as reveal the relation to more conventional models of NCGR that use Seiberg-Witten map as a main tool. With some modifications, reasoning from \cite{dobrski-ncgr} can be repeated\footnote{The most important difference is that since symplectic data are dynamical now, we cannot, at the level of action, assume that symplectic connection is flat, even in some special case.}. 

Within Fedosov formalism the Seiberg-Witten map turns out to be a local property of quantization of endomorphism bundle \cite{dobrski-sw}. More precisely, it appears as a relation between local star equivalences which trivialize star product of endomorphisms to star product of functions. Suppose that locally, in some frame $e$ in $\bund=T\sympman$, instead of $\nabla_i$ the flat connection is considered with vanishing connection coefficients. It gives rise to a star product of endomorphisms which, in this particular frame, is realized by the standard row-column rule of matrix multiplication, but with the usual product of functions replaced by Fedosov star product of functions $*_S$. This star product and our initial star product $*$ (the one relevant to the trace used in the action (\ref{ncactn})) are locally isomorphic. Let $M$ denote corresponding trivialization isomorphism. Notice, that the construction of $M$ is frame dependent due to the requirement of vanishing connection coefficients. It turns out that  $M'$ obtained for arbitrary different frame $e'=e \mathfrak{g}^{-1}$ is related to $M$ by Seiberg-Witten gauge transformation 
\begin{equation}
\label{sw-gauge}
M'(A_{(e')})=\widehat{\mathfrak{g}}*_S M(A_{(e)}) *_S \widehat{\mathfrak{g}}^{-1}
\end{equation}
where $A_{(e)}$ and $A_{(e')}$ are matrices corresponding to endomorphism $A$ in frames $e$ and $e'$ respectively, and hat on $\mathfrak{g}$ denotes Seiberg-Witten map obtained from Fedosov formalism (compare \cite{dobrski-sw} and also \cite{dobrski-genfed} for detailed analysis of such maps). Hence, trivialization isomorphisms appear as a Seiberg-Witten maps of endomorphisms, and Seiberg-Witten gauge transformations turn out to relate these isomorphisms. It should be stressed that such Seiberg-Witten maps are neither assumed nor appear as a solution to some postulated equation. They could be rather systematically computed form the very structure of Fedosov construction. This situation can be schematically visualized on the diagram. 
\begin{figure}[h]
\includegraphics[scale=0.8]{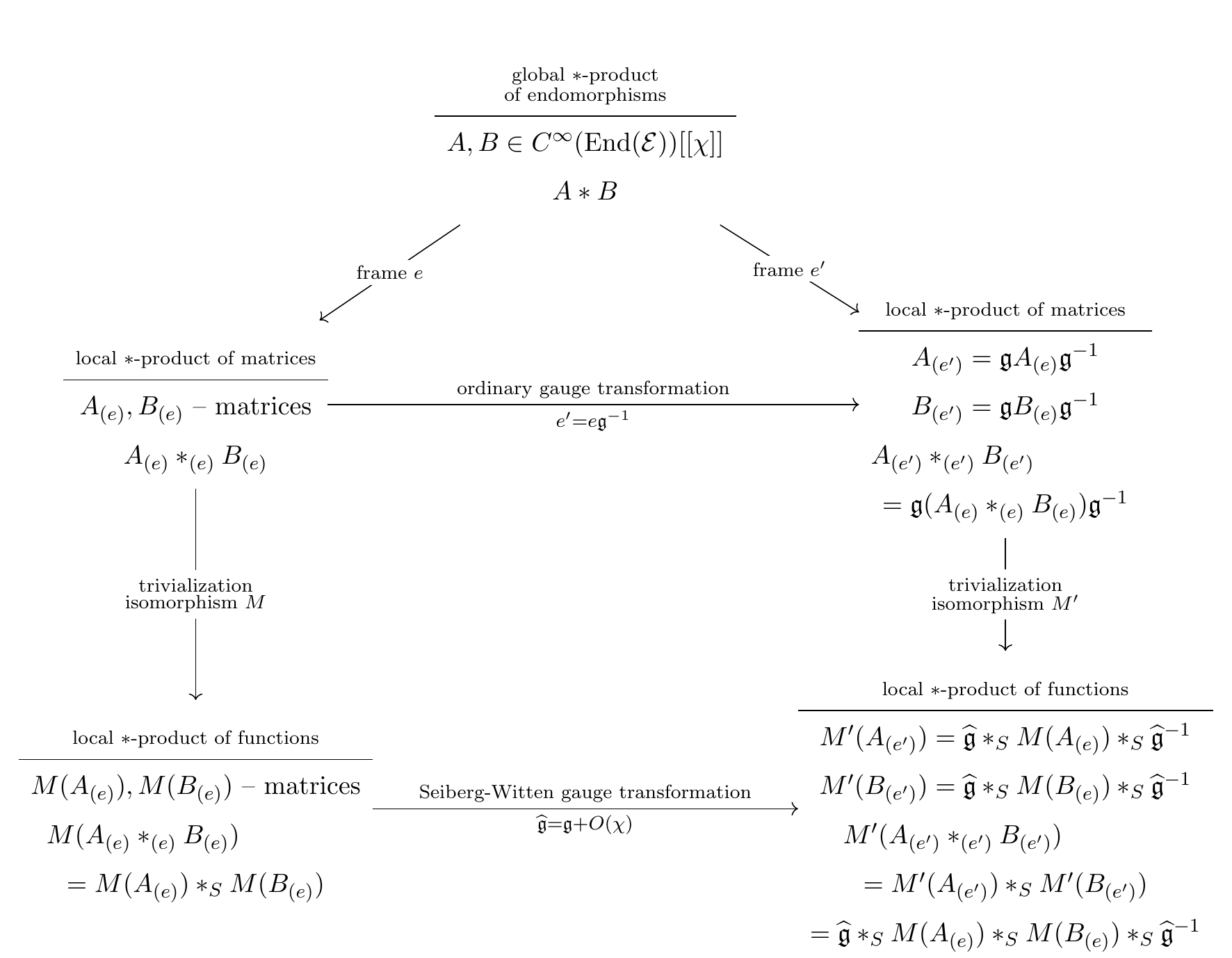}
\centering
\end{figure}

In the present case, if we assume that endomorphism $\underline{R}$ has the support small enough to be covered by a single frame, then, in virtue of (\ref{trinvprop}), one can rephrase (\ref{ncactn}) as
\begin{equation}
\label{act-local}
\widehat{S}=\frac{1}{16 \pi G} \tr_{*_S} (M(\underline{R}_{(e)})) 
\end{equation}
Such formulation, due to the trace property (\ref{trcprop}), is invariant to Seiberg-Witten gauge transformations (\ref{sw-gauge}) with $\mathfrak{g}$ and $\widehat{\mathfrak{g}}$ given by some local invertible sections of endomorphism bundle. Yet (\ref{sw-gauge}) is realized in terms of Fedosov star product of functions, i.e. it involves noncommutativity of the spacetime described by $*_S$. (Notice however, that in this setting large part of the dependence of the theory on the metric structure is transferred to the Seiberg-Witten map. In turn, the action (\ref{act-local}) is not manifestly coordinate covariant, although the resulting field equations can be of course brought to a covariant form equivalent to (\ref{fldeq})).

\subsection{Generalized Fedosov algebras}
Let us briefly comment on the option of using generalized variants of Fedosov construction, as described in \cite{dobrski-genfed}. The extension of original Fedosov formalism was considered there, originated in the possibility of founding the construction on a star product different from the Moyal one. (The Fedosov quantization can be viewed as a nontrivial geometrization of the Moyal quantization of a linear symplectic space). After providing results for fairly general prototypical products, the specific case of ``Moyal product with symmetric part of noncommutativity tensor'' was also considered. This approach refers to the geometrization of a variant of Moyal product with noncommutativity tensor given by $\theta^{ij}=\omega^{ij}+g^{ij}$, where $\omega^{ij}$ is symplectic, while $g^{ij}$ is symmetric and can be given the meaning of a metric. Two sub-cases was considered then -- the one with symplectic connection preserving symmetric part of noncommutativity tensor (i.e. metric), and the other, without such requirement. The first option is not suitable for the present purposes as it was clarified in the remark \ref{strongrmrk}. The second one leads to the formula for the trace functional, which, up to the first power of $\chi$, reads
\begin{equation}
\tr_* A = \int_\sympman \Tr \Bigg( 
A +  \frac{\rmi \chi}{2}\left( \curvbund_{ab} \Lambda^{ab} +  \frac{1}{2}\connsymp_b \connsymp_a g^{a b}\right) A +O(\chi^2) \Bigg) \frac{\omega^n}{n!}
\end{equation}
where $\curvbund_{ab}$ stands for the endomorphism defined by the curvature of the connection in $\bund$, and in our case is defined by the Riemann tensor. If we substitute $\underline{R}$ for $A$, then the term $\Tr\big(\curvbund_{ab} \underline{R} \Lambda^{ab} \big)=\tensor{R}{^i_j_a_b}\tensor{R}{^j_i} \Lambda^{ab}$ would vanish, exactly in the same way as in the generic case of (\ref{ncactn}). However, the term $(\connsymp_b \connsymp_a g^{a b})R$ introduces an imaginary quantity to the action integral, which is not cancelled by the weak compatibility condition. Indeed, under assumption $\tensor{Q}{^i_i_j}=0$, one can infer that $\connsymp_b \connsymp_a g^{a b}=-\connsymp_b (\tensor{Q}{^b_a_c}{g^{ac}})$, and the imaginary term still survives. This observation  prevents us from studying such theory here, although it also guides to considering additional compatiblilty condition $\tensor{Q}{^b_a_c}{g^{ac}}=0$ which is beyond the scope of the present paper.

\section{Comments}
\label{sect4}
The construction presented in this paper opens some possibilities for further research. First, it introduces nontrivial, yet apparently quite effective way of relating Fedosov and metric structures, by what we have called ``weak compatibilty condition''. It could be interesting to systematically summarize geometric consequences of this relation. 

Also, the new action principle for classical general relativity given by (\ref{graction}) and constraints (\ref{cnstr}) seems to be interesting on its own. In particular, it could be useful to analyze canonical formalism corresponding to it. Indeed, in a specific sense,  (\ref{graction}) and (\ref{cnstr}) are simpler then conventional Einstein-Hilbert action, as they do not contain any square root of determinant of the metric, and depend polynomially on fields and their inverses.

On the other hand, the weak compatibility is neither the only nor the canonical method for establishing dependence between metric and symplectic structures. It was already mentioned that some additional requirements can be analyzed within the framework of generalized Fedosov algebras. Furthermore, there exist other approaches that can be investigated. One possible example is given by \cite{yang}, where $g$ and $\omega$ are  related by imposing much stringent condition, formulated as the relation between metric and symplectic frame fields.

In the noncommutative case, although it was possible to remove Lagrange multipliers from the field equations, it remains unclear to what extent the theory determines symplectic connection and symplectic form. However, perturbative analysis (i.e. systematic expansion of fields in powers of $\chi$), analogous to that in \cite{dobrski-ncgr}, could become helpful here. Such considerations can be furthermore combined with inclusion of some (noncommutative) couplings to matter.

The related problem is the coupling to fermions  which is usually achieved in terms of frame (tetrad) fields. The starting point for such considerations could be section IV of \cite{dobrski-ncgr}, where a deformation of Palatini action by means of Fedosov construction (with fixed background of symplectic data) was considered.

Finally, let us mention two directions of research which are not inherently related to the present work, but have some area of common ground, and after some exploration could become strictly relevant to the presented model. The first one is given by the Ferraris-Francaviglia action principle for general relativity \cite{ferrfran}, which employs arbitrary auxiliary connection to bring the Einstein-Hilbert action to the manifestly and covariantly first-order form. This suggests, that symplectic connection can be used in this place. The other one is provided by the recent research \cite{asorey1,asorey2} on theories with symplectic connection coupled to the metric. Although it is driven by rather different paradigm then ours, it provides needful insight into quantum properties of such models.

\section{Acknowledgments}
I am very grateful to professor Maciej Przanowski for numerous discussions related to the present paper.
This work was supported by the Polish Ministry of Science and Higher Education grant no.~IP20120220072 within Iuventus Plus programme.

\section*{Appendix}
Here we give explicit, up to the second power of deformation parameter $\chi$, formula for auxiliary action integral of noncommutative theory. One can break it into three parts $\widehat{S}_{aux}=S+S_c+S_L$ where $S$ is action (\ref{graction}) without matter term, $S_c$ corresponds to noncommutative corrections and $S_L$ carries constraints. Using this split the respective terms can be inferred from the following expression
\begin{multline}
\widehat{S}_{aux}=S+S_c+S_L=\frac{1}{16 \pi G} \int_{\sympman} R \frac{\omega^n}{n!}  \\ 
+\frac{1}{16 \pi G} \int_{\sympman} \Big( - \frac{3}{8} \chi^2  \tensor{R}{^i_j} \tensor{R}{^j_{kmn}} \tensor{R}{^k_{irs}} \Lambda^{[mn}\Lambda^{rs]} 
+ 
 \chi^2 s_2 R  +  O(\chi^3) \Big) \frac{\omega^n}{n!} 
\\
+ \int_\sympman \left( \slgr^{ijk}\connsymp_i \omega_{jk} + \tensor{\clgr}{_i^{jk}} \tensor{\sGamma}{^i_{[jk]}} + \wlgr^i (\tensor{\Gamma}{^j_{ji}} -\tensor{\sGamma}{^j_{ji}}) \right) \frac{\omega^n}{n!}
\end{multline}
with
\begin{equation}
s_2=\frac{1}{64} \Lambda^{[a b} \Lambda^{c d]}\tensor{\curvsymp}{^k_{lab}}\tensor{\curvsymp}{^l_{kcd}}  + \frac{1}{48}\Lambda^{ab}\Lambda^{cd}\connsymp_e \connsymp_a \tensor{\curvsymp}{^e_{bcd}}
\end{equation}
where $\tensor{\curvsymp}{^a_{bcd}}$ is the curvature tensor of the symplectic connection $\connsymp$. The variation of these integrals with respect to $g_{ab}$, $\tensor{\sGamma}{^a_{bc}}$ and $\omega_{ab}$ yields (\ref{ncfldeq}) with

\begin{multline}
W^{ab} =-  \frac{\chi^2}{16 \pi G} \Bigg(\tensor{R}{^a^b} s_2 + \frac{3}{16} \tensor{R}{^b_c} \tensor{R}{^a^d_l_m} \tensor{R}{^c_d_n_r} \Lambda^{[l m}\Lambda^{n r]} + \frac{3}{16} \tensor{R}{^a_c} \tensor{R}{^b^d_l_m} \tensor{R}{^c_d_n_r} \Lambda^{[l m}\Lambda^{n r]} 
\\
- \tensor{g}{^a^c} \tensor{g}{^b^d} \nabla_d \nabla_c s_2 + \tensor{g}{^a^b} \tensor{g}{^c^d} \nabla_d \nabla_c s_2 + \frac{3}{16} \tensor{g}{^c^d} \nabla_d \nabla_c \tensor{R}{^a^l_m_n} \tensor{R}{^b_l_r_s} \Lambda^{[m n}\Lambda^{r s]} + \frac{3}{8} \nabla_d \nabla_n \tensor{R}{^c^d} \tensor{R}{^b_c_l_m} \Lambda^{[a l}\Lambda^{m n]} 
\\
- \frac{3}{8} \nabla_d \nabla_n \tensor{R}{^b^c} \tensor{R}{^d_c_l_m} \Lambda^{[a l}\Lambda^{m n]} + \frac{3}{8} \nabla_d \nabla_n \tensor{R}{^c^d} \tensor{R}{^a_c_l_m} \Lambda^{[b l}\Lambda^{m n]} - \frac{3}{8} \nabla_d \nabla_n \tensor{R}{^a^c} \tensor{R}{^d_c_l_m} \Lambda^{[b l}\Lambda^{m n]} 
\\
+ \frac{3}{32} \tensor{g}{^b^c} \tensor{g}{^d^l} \nabla_l \nabla_c \tensor{R}{^a_m_n_r} \tensor{R}{^m_d_s_t} \Lambda^{[n r}\Lambda^{s t]} + \frac{3}{32} \tensor{g}{^a^c} \tensor{g}{^d^l} \nabla_l \nabla_c \tensor{R}{^b_m_n_r} \tensor{R}{^m_d_s_t} \Lambda^{[n r}\Lambda^{s t]} 
\\
- \frac{3}{32} \tensor{g}{^b^c} \nabla_n \nabla_c \tensor{R}{^a^d_l_m} \tensor{R}{^n_d_r_s} \Lambda^{[l m}\Lambda^{r s]} 
- \frac{3}{16} \tensor{g}{^a^b} \nabla_n \nabla_d \tensor{R}{^c^d_l_m} \tensor{R}{^n_c_r_s} \Lambda^{[l m}\Lambda^{r s]}
 \\
+ \frac{3}{32} \tensor{g}{^a^c} \tensor{g}{^b^d} \nabla_r \nabla_c \tensor{R}{^l_d_m_n} \tensor{R}{^r_l_s_t} \Lambda^{[m n}\Lambda^{s t]}\Bigg)+O(\chi^3)
\end{multline}

\begin{multline}
\tensor{V}{_a^b^c} =  - \frac{\chi^2}{16 \pi G} \Bigg(\frac{1}{24}\tensor{\symped{R}}{^d_l_a_m} \Lambda^{b m} \Lambda^{c l} \connsymp_d R - \frac{1}{48}\tensor{\symped{R}}{^d_a_l_m} \Lambda^{b c} \Lambda^{l m} \connsymp_d R 
\\
- \frac{1}{48}\tensor{\delta}{_a^c} R \Lambda^{d l} \Lambda^{m n} \connsymp_l \tensor{\symped{R}}{^b_d_m_n} - \frac{1}{24}\Lambda^{b d} \Lambda^{c l} \connsymp_l \connsymp_d \connsymp_a R + \Lambda^{c d} \Lambda^{l m} \Bigg(-\frac{1}{48}\tensor{\symped{R}}{^b_d_l_m} \connsymp_a R 
\\
+ \frac{1}{48}\connsymp_d \Big( R \tensor{\symped{R}}{^b_a_l_m} \Big)+ \frac{1}{24} \connsymp_m \Big( R \tensor{\symped{R}}{^b_a_d_l}\Big)\Bigg) 
+ \Lambda^{b d} \Lambda^{l m} \Bigg(-\frac{1}{48}R \connsymp_a \tensor{\symped{R}}{^c_d_l_m} + \frac{1}{48}R \connsymp_d \tensor{\symped{R}}{^c_a_l_m}
\\
 + \frac{1}{24}R \connsymp_m \tensor{\symped{R}}{^c_l_a_d}\Bigg) \Bigg) +O(\chi^3)
\end{multline}

\begin{multline}
U^{ab}= \frac{\chi^2}{16 \pi G} \Bigg(\Bigg(\frac{15}{16}\tensor{R}{^c_d} \tensor{R}{^d_s_m_r} \tensor{R}{^s_c_l_n} + \frac{5}{128}R \tensor{\symped{R}}{^c_d_l_m} \tensor{\symped{R}}{^d_c_n_r}\Bigg) \Lambda^{[l n}\Lambda^{m r}\Lambda^{a b]} 
\\
+ \frac{1}{32} R \connsymp_n \connsymp_d \tensor{\symped{R}}{^n_c_l_m} \Lambda^{[a b}\Lambda^{c d]} \Lambda^{l m}  - \frac{1}{48} R \connsymp_n \connsymp_m \tensor{\symped{R}}{^n_l_c_d} \Lambda^{a c} \Lambda^{b d} \Lambda^{l m} \Bigg)+O(\chi^3)
\end{multline}
The calculations leading to above formulas were performed with the help of \texttt{xAct} tensor manipulation package \cite{xact}.

\end{document}